\documentclass[12pt,a4paper,reqno]{article}
\usepackage{graphics}
\usepackage{epsfig}

\usepackage{amsmath,amsthm,amssymb}
\newtheorem{theorem}{Theorem}[section]
\newtheorem{corollary}[theorem]{Corollary}
\newtheorem{lemma}[theorem]{Lemma}

\newtheorem{example}[theorem]{Example}
\theoremstyle{definition}
\newtheorem{definition}[theorem]{Definition}
\theoremstyle{remark}

\numberwithin{equation}{section}

\begin{document}
\title{ Exponent of  Cyclic Codes over $\mathbb{F}_{q}$}

\author{N. Annamalai\\
Assistant Professor\\
Indian Institute of Information Technology Kottayam\\
Pala, Kerala, India\\
{Email: annamalai@iiitkottayam.ac.in}
\bigskip\\
C Durairajan\\
Associate  Professor\\
Department of Mathematics\\ 
School of Mathematical Sciences\\
Bharathidasan University\\
Tiruchirappalli-620024, Tamil Nadu, India\\
{Email: cdurai66@rediffmail.com}
\hfill \\
\hfill \\
\hfill \\
\hfill \\
{\bf Proposed running head:} Exponent of a Cyclic Codes over $\mathbb{F}_{q}$}
\date{}
\maketitle

\newpage

\vspace*{0.5cm}
\begin{abstract} 
	In this article, we introduce and study the concept of the exponent of a cyclic code over a finite field $\mathbb{F}_q.$ We  give a relation between the exponent of a cyclic code and its dual code. Finally, we introduce and determine the exponent distribution of the cyclic code.
\end{abstract}
\vspace*{0.5cm}


\vspace{0.5cm}
\vspace{1.5cm}

\section{Introduction}
\quad Cyclic codes are a class of important linear codes and have generated great interest in
coding theory. In general, they have natural encoding and decoding algorithm. Since
cyclic codes can be described as ideals in polynomials residue rings, they have a rich
algebraic structure. There is a lot of literature on cyclic codes over fields and more
recently over rings. Cyclic codes have a unique property that is the codewords of cyclic
code can be divided into a numbers of mutually disjoint equivalence classes according
to an equivalence relation.

Let $\mathbb{F}_q$ be a finite field with $q$ elements. A code $C$ of length $n$ over $\mathbb{F}_q$ is a nonempty
subset of $\mathbb{F}_q^n.$ If $C$ is a subspace of $\mathbb{F}_q^n,$ then $C$ is called a $q$-ary linear code of length
$n.$ If the dimension of $C$ is $k,$ then the linear code $C$ is denoted by $[n, k]_q$ code. A
linear code $C$ of length $n$ over $\mathbb{F}_q$ is a cyclic code if for every $c =(c_0, c_1, \cdots, c_{n-1}) \in C,$
 $(c_{n-1}, c_0, \cdots, c_{n-2}) \in C.$
 
Let $R_n = \frac{\mathbb{F}_q[x]}{\langle x^n-1\rangle}$ be a polynomial residue ring. Then the following theorem gives a
relation between a cyclic code and an ideal of $R_n.$

\begin{theorem}\cite{vera}\label{thm0}
 The linear code $C$ is cyclic code if and only if $C$ is an ideal of $R_n=\frac{\mathbb{F}_q[x]}{\langle x^n-1\rangle}.$
\end{theorem}

 \begin{theorem}\cite{tod}\label{thm}
  Let $I$ be an ideal in $R_n=\frac{\mathbb{F}_q[x]}{\langle x^n-1\rangle}.$ Then
  \begin{enumerate}
   \item[1.] there is a unique monic polynomial $g(x)\in I$  of minimal degree,
   \item[2.] $I$ is principal with generator $g(x),$
   \item[3.] $g(x)$ divides $(x^n-1)$ in $\mathbb{F}_q[x].$
  \end{enumerate}
 \end{theorem}
 
 Note that $C$ is an ideal of $R_n$ and by Theorem \ref{thm}, $C$ is a principal ideal generated
by a unique monic polynomial $g(x) \in R_n.$ The polynomial $g(x)$ is called the generator
polynomial of the cyclic code $C.$ Since the generator polynomial $g(x)$ of $C$ is a divisor $x^n-1,$  $\frac{x^n-1}{g(x)}\in R_n.$



 Let $f(x)=a_n x^n+a_{n-1}x^{n-1}+\cdots+a_1 x+a_0\in \mathbb{F}_q[x]$ with $a_n\neq 0.$ Then the reciprocal
 polynomial $f^{*}$ of $f$ is defined by $$f^{*}(x)=x^n f\left(\frac{1}{x}\right)=a_0 x^n+a_1 x^{n-1}+\cdots+a_{n-1}x+a_n.$$
 
 Let $C$ be an $[n, k]_q$ code, then the dual code $C^{\perp}$ of the linear code $C \subseteq \mathbb{F}_q^n$ is defined by 
 $$C^{\perp}= \{ x \in \mathbb{F}_q^n \mid \langle x, c \rangle = 0\, \text{ for all } \, c \in C \}$$ where
    $\langle x, c \rangle = \sum_{i = 1}^n x_i c_i$ 
is a scalar product. Clearly, $C^{\perp}$ is an $[n, n - k]_q$ code.

\begin{theorem}\cite{lidl}\label{thm1}
 Let $g(x)$ be the generator polynomial of the cyclic code $C$ and let $h(x) =\frac{x^n-1}{g(x)}\in R_n.$ Then the dual code $C^{\perp}$ is cyclic with generator polynomial
$h^*(x).$
\end{theorem}

%
%
%
%
%
%

 Let $f(x)\in \mathbb{F}_q[x]$ be a nonzero polynomial. If $f(0)\neq 0$, then the least positive integer $e$ 
 for which $f(x)$ divides $x^e-1$ is called the {\it period of $f(x)$} or the {\it order of $f(x)$} and denoted by $ord(f)=ord(f(x)).$
 If $f(0)=0,$ then $f(x)=x^l g(x),$ where $l\in \mathbb{N}$ and $g(x)\in \mathbb{F}_q[x]$  with $g(0)\neq 0$ are uniquely
 determined and the $ord(f)$ is then defined by $ord(g).$

We state the followings for our further discussion.

\begin{lemma}\cite{lidl}\label{lem1}
Let $r$ be a positive integer. Then the polynomial $f(x)\in \mathbb{F}_q[x]$ with $f(0)\neq 0$ divides $x^r-1$ if and only if $ord(f)$ divides $r.$
\end{lemma}

\begin{corollary}\cite{lidl}\label{cor1}
 If $e_1$ and $e_2$ are positive integers, then the greatest common divisor of $x^{e_1}-1$ and $x^{e_2}-1$
 in $\mathbb{F}_q[x]$ is $x^d-1$,  where $d$ is the greatest common divisor of $e_1$ and $e_2.$
\end{corollary}

\begin{theorem}\cite{lidl}\label{thm2}
 Let $f$ be a nonzero polynomial in $\mathbb{F}_q[x]$ and $f^{*}$ its reciprocal polynomial. Then $ord(f)=ord(f^{*}).$
\end{theorem}

\begin{theorem}\cite{lidl}\label{thm3}
 Let $g_1, g_2, \cdots, g_k$ be pairwise prime nonzero polynomials over $\mathbb{F}_q$ and let $f=g_1 g_2\cdots g_k.$ Then $ord(f)=lcm(ord(g_1), ord(g_2), \cdots, ord(g_k)).$
\end{theorem}

In this article, we introduced the concept of the exponent of a cyclic code over $\mathbb{F}_q$ and
studied the exponent of the dual cyclic code in section 3. We introduced and determined
the exponent distribution of the cyclic code over $\mathbb{F}_q$ in section 4.

\section{Exponent of  Cyclic Codes over $\mathbb{F}_q$}
\quad In this section, we define the exponent of a cyclic code over $\mathbb{F}_q$ and discussed its properties.

 \begin{definition}
  Let $C$ be a cyclic code of length $n$ over $\mathbb{F}_q$ with generator polynomial $g(x).$ Then the least positive integer $e$ for
  which $g(x)$ divides $x^e-1$ is called the {\it exponent of the cyclic code $C$} and denoted by $e=exp(C).$
 \end{definition}
 Clearly, the exponent of $C$ is the same as the order of its generator polynomial.
 
\begin{example}\label{ex1}
 Let $C=\{000, 101, 110, 011\}$ be a cyclic code of length 3 over $\mathbb{F}_2.$ Then the generator polynomial of $C$ is $1+x$ and $1+x\mid x^1-1.$ 
 Hence $exp(C)=1.$
\end{example}

\begin{theorem}\label{thm4}
 Let $C$ be an $[n, k]_q$ cyclic code over $\mathbb{F}_q$ and let $e=exp(C).$ Then $e\mid n.$ 
\end{theorem}

\begin{proof}
 Let $C$ be an $[n, k]_q$ cyclic code over $\mathbb{F}_q$ with exponent $e$ and let $g(x)$ be its generator polynomial. Then $g(x)\mid x^n-1$ and hence
 by Lemma \ref{lem1}, $ord(g)\mid n.$ Since $e = exp(C)$ and $ord(g) = exp(C),$ hence $e \mid n.$
\end{proof}
Note that by the above theorem, the exponent of a cyclic code is always less than or
equal to its length. The bound is reached by the following example.
\begin{example}
 Let $C=\{000, 111, 222\}$ be a cyclic code of length 3 over $\mathbb{F}_3.$ Then the generator polynomial $ 1+x +x^2$ of $C$ divides $x^3-1$ and hence $exp(C)=3.$
\end{example}

\begin{theorem}\label{thm5}
 Let $C$ be an $[n, k]_q$ cyclic code over $\mathbb{F}_q.$ If $e=exp(C),$ then $n-k\leq e.$ 
\end{theorem}

\begin{proof}
 Let $C$ be an $[n, k]_q$ cyclic code over  $\mathbb{F}_q$ and let $g(x)$ be its generator polynomial. Then $deg(g(x))=n-k.$
 If $e=exp(C),$ then $g(x)| x^e-1.$ This implies that, $deg(g(x))\leq deg(x^e-1).$ That is, $n-k\leq e.$ 
\end{proof}
By Theorems \ref{thm4} and \ref{thm5}, we have

\begin{corollary}\label{cor2}
 Let $C$ be an $[n, k]_q$ cyclic code over $\mathbb{F}_q.$ If $e=exp(C),$ then $n-k\leq e \leq n.$
\end{corollary}
The lower bound of the above corollary is reched by the Example \ref{ex1}
 
\begin{theorem}\label{thm6}
 Let $C$ be an $[n, k]_q$ cyclic code over $\mathbb{F}_q$ and let $g(x), h(x)$ be the generator and parity-check polynomials of $C,$ respectively. If $(g(x), h(x))=1,$ 
 $exp(C)=e$ and $r=exp(C^{\perp}),$ then $n=lcm(e, r).$ 
\end{theorem}

\begin{proof}
 Let $g(x)$ be a generator polynomial and let $h(x)$ be a parity check polynomial of $C.$ Since $x^n-1=g(x)h(x)$ and $(g(x), h(x))=1,$ by Theorem \ref{thm3}, we have $n=lcm(e, r).$ 
\end{proof}
\begin{theorem}
  Let $C_i$ be an $[n, k_i]_q$ cyclic code over $\mathbb{F}_q$ for $i=1, 2.$ If $e_i$ is the exponent of $C_i,$ then $exp(C_1 + C_2)=gcd(e_1, e_2).$
 \end{theorem}
\begin{proof}
 Let $C_1, C_2$ be two cyclic codes of length $n$ over $\mathbb{F}_q.$ Then
 $$C_1+ C_2=\{c_1+ c_2\mid c_i\in C_i, i=1, 2\}$$
 is a cyclic code of length $n$ over  $\mathbb{F}_q.$
 Let $g_1(x), g_2(x), g(x)$ be 
 the generator polynomial of $C_1, C_2$ and $C_1 + C_2,$ respectively and let $exp(C_1+ C_2)=e.$ Then $g(x)\mid x^e-1.$ 
 Since $C_i=\langle g_i(x)\rangle,$  $C_1+ C_2=\langle gcd(g_1(x), g_2(x)\rangle.$ This implies that, $g(x)\mid g_i(x)$ for $i=1,2.$
 Then $e$ divides $e_i$ for $i=1, 2.$ Hence $e$ divides $gcd(e_1, e_2)=d.$ That is, $e \leq d.$
 
 Suppose that $e < d.$ Then by Division Algorithm, there exist $q, r \in \mathbb{Z}$ such that $d=eq+r$ where $0\leq r<e.$ Consider
 \begin{align*}
  x^d-1&=x^{eq+r}-1\\
  &=x^{eq}x^r-x^r +x^r-1\\
  &=x^r(x^{eq}-1)+x^r-1\\
  (x^d-1)-x^r(x^{eq}-1)&=x^r-1.
 \end{align*}
Since $g(x)$ divides both $x^e-1$ and $x^d-1,$  $g(x)$ divides $x^r-1.$ By Lemma \ref{thm1}, we have $e\leq r,$ a contradiction to $r < e.$  Therefore, $e < d$ is impossible and hence $e=d=gcd(e_1, e_2).$
 \end{proof}

\section{Exponent Distribution of a Cyclic Codes over $\mathbb{F}_q$}
\quad In this section, we introduce the exponent distribution of a cyclic code over $\mathbb{F}_q$ and discuss its properties.

We know that, if $c=(c_0, c_1, \cdots, c_{n-1})$ is a codeword in the cyclic code $C,$ then  its corresponding polynomial representation is $c(x)=c_0+c_1 x+\cdots+c_{n-1}x^{n-1}.$ Since $c(x) \in C = \langle
 g(x) \rangle,$ there exists a polynomial $a(x)\in R_n$ such that $c(x)=a(x)g(x).$ 
We call $a(x)$ is the information polynomial of $c.$

For $1\leq t\leq n,$ we define  $A_t:= \{c(x) \in C \mid ord(c(x))=t\}.$ Let $B_t=\mid A_t\mid,$ then the sequence $\{B_t\}$ is called the {\it exponent distribution} of
the cyclic code $C$ over $\mathbb{F}_q.$

\begin{example}
 Let $C=\{000, 101, 110, 011\}$ be a cyclic code over $\mathbb{F}_2.$ Then the corresponding codeword
  polynomials are $0, 1+x^2, 1+x, x+x^2$ and $ord(1+x^2)=2, ord(x+x^2)=1, ord(1+x)=1.$ This implies that, $A_1=\{1+x, x+x^2\}, A_2=\{1+x^2\}$ and $ A_3=\emptyset.$
  Hence $(2, 1, 0)$ is the exponent distribution of the cyclic code $C.$
\end{example}

\begin{example}
Let $C=\{000, 111, 222\}$ be a cyclic code over $\mathbb{F}_3.$ Then the corresponding codeword
  polynomials are $0, 1+x+x^2, 2+2x+2x^2$ and $ord(1+x+x^2)=3, ord(2+2x+2x^2)=3.$ This implies that, $A_1=\emptyset, A_2=\emptyset, A_3=\{1+x+x^2, 2+2x+2x^2\}.$
    Hence $(0, 0, 2)$ is the exponent distribution of the cyclic code $C.$
\end{example}

\begin{example}
 Let $C=\{000, 100, 010, 001, 110, 011, 101, 111\}$ be a cyclic code over $\mathbb{F}_2.$ Then the corresponding codeword
  polynomials are $0, 1, x, x^2, 1+x, 1+x^2, x+x^2, 1+x+x^2.$ This implies that, $A_1=\{1, x, x^2, 1+x\}, A_2=\{1+x^2\}, A_3=\{1+x+x^2\}.$
    Hence $(4, 1, 1)$ is the exponent distribution of the cyclic code $C.$
\end{example}

\begin{theorem}
 Let $C$ be an $[n, k]_q$ cyclic code over $\mathbb{F}_q$ with generator polynomial $g(x).$ If $G$ is a generator matrix of $C,$ then the order of the polynomials 
 corresponding to a basis of $C$ are same. 
\end{theorem}
\begin{proof}
Let $g(x)$ be a generator polynomial of $C.$ Then $\{g(x), xg(x), x^2g(x), \cdots, x^{k-1}g(x)\}$ is a basis of $C.$ Since $ord(g(x)) = ord(x^ig(x))$  for all $i \geq 0$ and hence  the order of the polynomials 
corresponding to a basis of $C$ are same.
\end{proof}
 \begin{corollary}
 Let $C$ be an $[n, k]_q$ cyclic code over $\mathbb{F}_q$ with generator polynomial $g(x).$ If $ord(g(x))= e,$ then $B_e \geq k.$
 \end{corollary}
\begin{proof} Given that  $ord(g(x))= e.$ By the above theorem, there are $k$ elements in $A_e$ and hence $B_e \geq k.$
\end{proof}


\section{Conclusion}
\quad In this paper, we have introduced and studied the concept of the exponent of a cyclic code over a finite field $\mathbb{F}_q.$ We have given a relation between the exponent of a cyclic code and its dual code. Finally, we introduced determined the exponent distribution of the cyclic code. A future work is to find bounds for the number of cyclic codes of given length and exponent over $\mathbb{F}_q.$ Finding the exponents of the BCH, Reed-Soloman and Goppa codes are other
direction to work futher.

\end{document}